\documentclass[11pt]{article}

\usepackage{tikz}
\usepackage{verbatim}
\usepackage{cancel}
\usepackage{enumitem}
\usepackage{bm}
\usepackage{caption}
\usepackage{changepage}
\usepackage{mathtools}
\usepackage{bbm}
\usepackage{times}

\usepackage[dvipsnames]{xcolor}

\usepackage{sn-preamble}
\usepackage{bbm}
\usepackage{natbib}

\newcommand{\red}[1]{{{#1}}}

\newcommand{\ignore}[1]{}

\newcommand{\train}{\mathbf{S}_{\textsf{train}}}
\newcommand{\ltrain}{\overline{\mathbf{S}}_{\textsf{train}}}
\newcommand{\ldtrain}{\overline{\mathcal{D}}_{\textsf{train}}}
\newcommand{\test}{\mathbf{X}_{\textsf{test}}}
\newcommand{\acc}{\textsf{ACCEPT}}
\newcommand{\rej}{\textsf{REJECT}}

\newcommand{\dtrain}{\calD^{\textsf{train}}}
\newcommand{\dtest}{\calD^{\textsf{test}}}

\newcommand{\tdsboost}{\textsc{TDSboost}}

\newcommand{\frozen}{\mathsf{frozen}}

\newcommand{\pmin}{p_{\mathrm{min}}}
\newcommand{\lbl}{\mathsf{label}}
\newcommand{\weakdist}{\mathsf{WD}}
\newcommand{\getweakdist}{\textsc{GetWeakDistinguisher}}
\newcommand{\amin}{a_{\mathrm{min}}}
\renewcommand{\next}{\mathsf{next}}
\newcommand{\Dtrain}{\calD_1}
\newcommand{\Dtest}{\calD_0}
\newcommand{\accuracy}{\eps}
\newcommand{\rejrate}{\eta}
\newcommand{\leaves}{\mathsf{leaves}}

\newcommand{\dtrainlabeled}{\bar{\calD}^{\textsf{train}}}
\newcommand{\dtestlabeled}{\bar{\calD}^{\textsf{test}}}

\DeclareMathOperator*{\argmin}{arg\,min}

\renewcommand{\bw}{w}
\newcommand{\sgn}{\mathrm{sign}}

\newcommand{\learnh}{\textsc{Learn-Halfspace}}
\newcommand{\forster}{\textsc{Forster}}
\newcommand{\learnmargin}{\textsc{Learn-High-Margin-Halfspace}}

\definecolor{comments}{RGB}{80,120,120}

\usepackage{enumitem}

\title{Equivalence of Coarse and Fine-Grained Models \\ for Learning with Distribution Shift}

\author{%
    \begin{tabular}{cc}
        \begin{tabular}{c}
            Adam R. Klivans\thanks{Supported by the NSF AI Institute for Foundations of Machine Learning (IFML).}\\ \texttt{klivans@cs.utexas.edu} \\ UT Austin 
        \end{tabular} & 
        \begin{tabular}{c}
            Shyamal Patel\thanks{Supported by NSF awards CCF-2106429, CCF-2107187, CCF-2218677, ONR award
ONR-13533312, and an NSF Graduate Student Fellowship.} \\ \texttt{shyamalpatelb@gmail.com} \\ Columbia University
        \end{tabular}
        \\\\
        \begin{tabular}{c}
             Konstantinos Stavropoulos\thanks{Supported by the NSF AI Institute for Foundations of Machine Learning (IFML), and by the 2025 Apple Scholars in AI/ML PhD fellowship.} \\ \texttt{kstavrop@cs.utexas.edu} \\ UT Austin
        \end{tabular}
         & 
         \begin{tabular}{c}
              Arsen Vasilyan\thanks{Supported by the NSF AI Institute for Foundations of Machine Learning (IFML).} \\ \texttt{arsenvasilyan@gmail.com} \\ UT Austin
         \end{tabular}
    \end{tabular}
}

\date{}

\begin{document}

\maketitle

\begin{abstract}
Recent work on provably efficient algorithms for learning with distribution shift has focused on two models: PQ learning \cite{goldwasser2020beyond} and TDS learning \cite{klivans2024testable}.  Algorithms for TDS learning are allowed to reject a test set entirely if distribution shift is detected.  In contrast, PQ learners may only reject points that are deemed out-of-distribution on an individual basis.  Our main result is a surprising equivalence between these two models in the distribution-free setting.  In particular, we give an efficient black-box reduction from PQ learning to TDS learning for any Boolean concept class.  This equivalence implies the first hardness results for distribution-free TDS learning of basic classes such as halfspaces.  The main technical contribution underlying our equivalence is a method for boosting, via branching programs,  the weak distinguishing power of TDS learners that have rejected the target domain. 

We also show that giving a learner access to {\em membership queries} sidesteps these hardness results and allows for efficient, distribution-free PQ learnability of halfspaces.  Our algorithm iteratively recovers large-margin separators obtained by applying successive Forster transforms on the training data.

\end{abstract}

\newpage

\section{Introduction}

One of the most fundamental models for learning with distribution shift is \emph{domain adaptation}, where a learner is given access to a labeled training set drawn from a source distribution ${\cal D}$, an unlabeled test set drawn from a potentially different target distribution ${\cal D'}$, and the goal is to output a classifier that performs well with respect to the test set.   The relationship between ${\cal D}$ and ${\cal D'}$ may be arbitrary, and prior work established conditions under which a classifier trained on some source distribution will generalize to a different target distribution
\citep{ben2006analysis,blitzer2007learning,mansour2009domadapt,ben2010theory,david2010impossibility}.    In particular, these works prove generalization bounds assuming that some notion of distance between the source and target domains is bounded (e.g., discrepancy distance).  These distances, however, are not efficiently verifiable and do not give algorithmic guarantees.  

Motivated by these concerns, Goldwasser et al.~\cite{goldwasser2020beyond} defined the model of PQ learning, where a learner produces a classifier that may abstain on points from the test domain that it believes are out-of-distribution relative to the source domain.  For points that are in-distribution the classifier is guaranteed to achieve low error.  They give an algorithm that succeeds even under potentially extreme distribution shift but requires access to an oracle for empirical risk minimization (ERM).  It is well known, however, that performing ERM is NP-hard for very simple function classes.  In fact, their reduction requires an ERM oracle even when the training distribution is nicely structured (e.g., a Gaussian). 

In search of truly efficient algorithms, Klivans et al.~\cite{klivans2024testable} introduced the model of Testable Distribution Shift learning (TDS learning), where an algorithm may abstain at the \emph{population level}, either outputting a classifier with low error on the test set or rejecting the test set entirely. The rejection criterion of TDS learning is therefore much coarser than in PQ learning, where rejection is performed individually on test points. \cite{klivans2024testable} gave the first efficient algorithms for learning well-studied concept classes in this model (such as halfspaces), whenever the training distribution is nicely structured (e.g., log-concave or uniform on the hypercube).  Their algorithms make no assumptions on the test distribution and do not require oracles for ERM.  Extensions and improvements on these initial TDS results have appeared in multiple recent works \citep{klivans2024testable,klivans2024learning,chandrasekaran2024efficient,goel2024tolerant,chandrasekaran2025learning}. 

\subsection{Our Main Result}
  In this paper, we prove, surprisingly, that distribution-free TDS learning and PQ learning are equivalent.  Our main result is a computationally efficient black-box reduction from PQ learning to TDS learning, which shows that a learning algorithm that rejects distributions at the population level can be converted into one that makes per-example decisions about whether test points are in or out of distribution.  At a high level, our reduction uses the method of boosting by branching programs, viewing TDS learners that reject as weak distinguishers of samples drawn from source and target domains. 

This equivalence has strong implications: it immediately gives the first set of hardness results for distribution-free TDS learning of halfspaces (even halfspaces with a margin) and other basic concept classes, as these are known to be computationally intractable in the PQ setting \citep{kalai2021efficient}.  It further shows that assumptions on the training distribution are {\em necessary} for efficient TDS learnability, even for simple concept classes.

\subsection{PQ Learning with Queries}
Given these negative results, a natural question is whether there is some other way to circumvent hardness.  As mentioned above, prior work on TDS learning has shown that assumptions on the training distribution such as log-concavity yield positive results for TDS learning and in some cases PQ learning \citep{klivans2024testable,klivans2024learning,chandrasekaran2024efficient,goel2024tolerant,chandrasekaran2025learning}.

Here we explore a different path: can we design efficient, distribution-free PQ learning algorithms by giving the learner \emph{membership query access}.  That is, the learner can query the labels of arbitrary points in the support of the training distribution before seeing the test set.  We make no assumptions on either the training or the test distributions.  This setup provides more power to the learner but still preserves the core difficulty of distribution shift. 

Recent work by \cite{lange2025limitations} shows that membership queries do {\em not} help in the related model of testable learning, suggesting that query access may not be helpful in our setting.  
Perhaps unexpectedly, we show that queries do help in this setting and present an efficient, distribution-free PQ learner for halfspaces.  Our algorithm works in stages, first using Forster transform to convert the halfspace into one with large enough margin, and then using queries to recover the parameters of the transformed halfspace.
These results demonstrate a clear separation between testable learning and TDS learning. While membership queries fail to help in testable learning, they can circumvent hardness in the presence of distribution shift.

\subsection{Technical Overview}

\paragraph{Reducing PQ Learning to TDS Learning.}
Before describing our reduction, we briefly review the TDS and PQ learning models (see \Cref{definition:tds-learning} and \Cref{definition:pq-learning}). In both models, the learner has access to labeled samples from a training distribution $\dtrain$ and unlabeled samples from a test distribution $\dtest$.

\begin{itemize}
    \item A TDS learner either accepts or rejects. Upon acceptance, it must output a hypothesis whose error on $\dtest$ is at most $\eps$. In addition, the learner is required to accept with high probability in the case where $\dtrain = \dtest$.
    
    \item A PQ learner outputs a hypothesis $h$ together with a selector $g$, which determines whether a given point should be classified. The selector must reject at most an $\eps$ fraction of points drawn from $\dtrain$. Furthermore, for a random $\bx \sim \dtest$, the probability that $g(\bx)$ accepts while $h(\bx)$ misclassifies $\bx$ must be at most $\eps$.
\end{itemize}

Suppose we are given a TDS learning algorithm $\calA$. If we run $\calA$ on samples from $\dtrain$ and $\dtest$ and it accepts, then with high probability the hypothesis $h$ it outputs has error at most $\eps$ on $\dtest$. In this case, we immediately obtain the guarantees of PQ learning by taking the selector $g \equiv 1$ (i.e., rejecting no test points) together with the hypothesis $h$.

If, on the other hand, $\calA$ rejects, then it distinguishes between $\dtrain$ and $\dtest$ using $m$ samples. In this regime, we may use a hybrid argument to obtain a weak distinguisher $\weakdist$ that, given a single sample, can distinguish whether it was drawn from $\dtrain$ or $\dtest$ with advantage $\Omega\!\left(\frac{1}{m}\right)$. Our main objective is then to boost this weak distinguishing advantage, producing an algorithm that still operates on a single sample but distinguishes between $\dtrain$ and $\dtest$ with high probability.

Before describing how we do this, we start by explaining why a black box boosting algorithm will not be sufficient for our purposes. In particular, consider using an algorithm such as AdaBoost to accomplish our goal. Running the AdaBoost algorithm with our weak distinguishing algorithm $\weakdist$ results in new distributions $(\dtrain)'$ and $(\dtest)'$. We can then run $\calA$ on samples from $(\dtrain)'$ and $(\dtest)'$. We have the following cases:
\begin{itemize}
    \item If $\calA$ rejects, we get a weak distinguisher between these two distributions and can continue with the boosting procedure.
    \item However, if $\calA$ accepts, we get a hypothesis with good accuracy under $(\dtest)'$ rather than a distinguisher. Unfortunately, $h$ could have large error under $\dtest$ as $(\dtest)'$ could be significantly different. Moreover, since we don't have a weak distinguisher between $(\dtrain)'$ and $(\dtest)'$, it is unclear how to continue to make progress towards distinguishing $\dtrain$ and $\dtest$.
\end{itemize}

To circumvent this issue, we will show that a modification of the martingale boosting algorithm of \cite{long2005martingale} will allows us to prove our result. In particular, using this scheme, we will partition the domain into disjoint regions $\calR_1, \dots, \calR_k$ such that for each region $\calR_i$ either $(i)$ $\calA$ accepts on $\dtrain|_{\calR_i}$ and $\dtest|_{\calR_i}$ where $\cdot|_{\calR_i}$ denotes the distribution conditioned on $x \in \calR_i$, $(ii)$ points in $\calR_i$ are primarily from $\dtest$, i.e. $\Pr_{\bx \sim \dtrain}[\bx \in \calR_i] \leq \eps \Pr_{\dtest}[\bx \in \calR_i]$, or $(iii)$ points in $\calR_i$ are primarily from $\dtrain$, meaning $\Pr_{\dtest}[\bx \in \calR_i] \leq \eps \Pr_{\bx \sim \dtrain}[\bx \in \calR_i]$.

Notably, given such a decomposition, we can set our selector function to reject on regions satisfying $(ii)$, i.e. those containing many more points from $\dtest$ than $\dtrain$. In particular,
\[g(x) = \begin{cases}
    0 & \exists \calR_i: x \in \calR_i \land \Pr_{\bz \sim \dtrain}[\bz \in \calR_i] \leq \eps \Pr_{\bz \sim \dtest}[\bz \in \calR_i] \\
    1 & \text{otherwise}
\end{cases}\]

Notably, such a selector on rejects a small fraction of $\dtrain$ as
    \[\Pr_{\bx \sim \dtrain}[g(\bx) = 0] = \!\! \sum_{\substack{i \in [k]: \\ \calR_i \text{ satisfies } (ii)}} \!\!\!\!\! \Pr_{\bx \sim \dtrain} [\bx \in \calR_i] \leq \!\! \sum_{\substack{i \in [k]: \\ \calR_i \text{ satisfies } (ii)}} \!\!\!\!\! \eps \Pr_{\bx \sim \dtest} [\bx \in \calR_i] \leq \eps \]

With such a selector, we can then set our hypothesis $h(x)$ as follows: If $x$ is contained in a region $\calR_i$ satisfying $(i)$ set $h(x) := h_i(x)$, where $h_i$ is the hypothesis output by $\calA$ on $\dtrain|_{\calR_i}$ and $\dtest|_{\calR_i}$. On the other hand, for all $x$ in regions $\calR_i$ satisfying $(iii)$ we set $h(x) = 1$.

By an analogous computation as done for the selector, we have that points from $\dtest$ lie in a region $\calR_i$ satisfying $(iii)$ with probability at most $\eps$. On the other hand, each $h_i$ has error at most $\eps$ under $\dtest|_{\calR_i}$. Since the selector rejects any point lying in a region satisfying $(ii)$, it then follows that the probability that $g(x) = 1$ and $h$ makes an error is then at most $2 \eps$, yielding the desired PQ learner.

\paragraph{PQ Learning Halfspaces with Queries.} 
To prove our result on PQ learning halfspaces, we start by noting that it essentially suffices to prove the statement for homogenous halfspaces, i.e. those of the form $\sign(w \cdot x)$ with $w \in \R^n$. In this homogenous setting, we'll assume, without loss of generality, that $\dtrain, \dtest \subseteq \mathbb{S}^{n-1}$ and that the vector $w \in \mathbb{R}^n$ defining the halfspace satisfies $\|w\| = 1$.

The key idea is then to simply notice that if $w \cdot x$ has many points with margin at least $\gamma$, i.e. satisfying $|w \cdot x| \geq \gamma$, then any halfspace $\sgn(\wh{w} \cdot x)$, where $\wh{w}$ has Euclidean distance at most $\gamma/3$ from $w$, agrees with $w$ on these high margin points. In particular, any point $x \in \mathbb{S}^{n-1}$ with $|\wh{w} \cdot x| \geq 2\gamma/3$ has $\sgn(\wh{w} \cdot x) = \sgn(w \cdot x)$. Moreover, any point $x \in \mathbb{S}^{n-1}$ with $|w \cdot x| \geq \gamma$ must satisfy $|\wh{w} \cdot x| \geq 2\gamma/3$. Thus, if we could compute a high accuracy estimate $\wh{w}$ of $w$, we could use the selector function $g(x) = \mathbb{I} \{|\wh{w} \cdot x| \geq 2\gamma/3\}$ and hypothesis $h(x) = \sgn(\wh{w} \cdot x)$ to correctly label all high margin points. On the other hand, computing such an estimate $\wh{w}$ can be done easily with queries by noting that the center of mass under $\calN(0,I_n)$ of the points $\{x \in \R^n: w \cdot x \geq 0\}$ points in the direction of $w$.

Of course, since the distribution $\dtrain$ is arbitrary, it's possible that no points have large margin. As such, our algorithm first applies a Forster transform to transform the point set into one where there is a non-trivial fraction of points with a margin. We can then run the above high margin PQ learning algorithm to correctly label all points with a high margin. Afterwards, we can recurse on the unlabeled points until only an $\eps$ fraction of points from $\dtrain$ are unlabeled.

Given a point $\bx \sim \dtest$, we apply the transformations and classifiers computed for $\dtrain$ to $\bx$. If any of the high margin PQ learners that we trained accepts, we accept and label the point according to the learner. Otherwise, we reject. Notably, since our high margin PQ learner correctly labels all points accepted by its selector, our hypothesis will have low error over $\dtest$ whenever the selector accepts. On the other hand, since our transformations and classifiers label most of our training set, we expect only an $\eps$ fraction of points from $\dtrain$ to be rejected by our selector, as desired.

\subsection{Related Work} 

\paragraph{Learning via Refutation}
Recall that the first step of our proof is to obtain a weak distinguisher between the training and test distributions using a hybrid argument. Similar hybrid arguments have appeared in prior work on learning via refutation \citep{vadhan2017learning,kothari2018improper}. In particular, \cite{kothari2018improper} use such an argument to derive a weak PAC learner from a distinguishing algorithm that can tell whether (i) there exists a concept in a given class achieving nontrivial advantage on a labeled distribution, or (ii) the labels are completely random. This weak learner is then converted into a strong agnostic learner via a black-box boosting procedure.

In contrast, standard black-box boosting does not apply in our setting, since the TDS learner may accept and output a hypothesis rather than explicitly distinguish the two distributions. To address this, we build on the boosting framework of \cite{long2005martingale} and develop a new analysis that accommodates weak distinguishers that are allowed to abstain.

\paragraph{Domain Adaptation}
Domain adaptation has been extensively studied over the past two decades \citep{ben2006analysis,blitzer2007learning,mansour2009domadapt,ben2010theory,david2010impossibility,redko2020survey}. As in our setting, this line of work considers learners that observe labeled samples from a source distribution and unlabeled samples from a target distribution, and aim to achieve low error on the target distribution—without the option to reject.

A classical approach \citep{ben2006analysis,blitzer2007learning,mansour2009domadapt} bounds the target error of an empirical risk minimizer trained on source data by the sum of a joint optimal error term and a measure of divergence between the source and target marginals, such as the discrepancy or $d_A$ distance. These divergences can be estimated statistically from unlabeled data, yielding statistically efficient guarantees. However, computing discrepancy-based distances is computationally intractable in general, requiring exponential time even for simple concept classes such as halfspaces \citep{chandrasekaran2024efficient}.

\paragraph{PQ Learning}
PQ learning was introduced by \cite{goldwasser2020beyond}, who showed that access to an agnostic ERM oracle suffices to solve the problem. \cite{kalai2021efficient} subsequently established that PQ learning is computationally equivalent to \emph{reliable agnostic learning} \cite{kalai2012reliable}. Reliable agnostic learning is strictly easier than standard agnostic learning; for example, parities over the Boolean hypercube admit efficient reliable agnostic learners. In contrast, reliably learning conjunctions—and hence PQ learning for this class—is at least as hard as PAC learning DNFs, a long-standing open problem that serves as a canonical barrier in computational learning theory.

\cite{goel2024tolerant} later studied a distribution-specific variant of PQ learning and obtained the first efficient algorithms for several concept classes under Gaussian training distributions.

\paragraph{Testable Learning}
Testable learning with distribution shift (TDS learning) was introduced by \cite{klivans2024testable} and is closely related to the model of testable agnostic learning of \cite{rubinfeld2022testing}. Both frameworks require the learner to certify its own performance. In testable agnostic learning, there is no distribution shift, and the learner must either output a hypothesis with near-optimal error or reject upon detecting a violation of a target distributional assumption. TDS learning generalizes this framework to settings where the training and test distributions may differ.

A substantial line of work has developed computationally efficient algorithms for testable agnostic learning \citep{rubinfeld2022testing,gollakota2022moment,gollakota2023efficient,diakonikolas2024testable,gollakota2024tester,diakonikolas2023efficient,slot2024testably,klivans2025the} and for TDS learning \citep{klivans2024testable,klivans2024learning,chandrasekaran2024efficient,goel2024tolerant,chandrasekaran2025learning}, primarily in distribution-specific settings such as the Gaussian distribution or the uniform distribution over the Boolean hypercube.

In the distribution-free setting, testable agnostic learning coincides with standard agnostic learning. Prior to our work, the strongest known relationship was that TDS learning is no harder than PQ learning, and hence no harder than reliable agnostic learning. In particular, it was unknown whether TDS learning of halfspaces admits an efficient algorithm, even in the presence of margin assumptions. Our results resolve this question by showing that TDS learning is computationally equivalent to reliable agnostic learning: strictly easier than agnostic learning, yet harder than realizable PAC learning. As a consequence, TDS learning of halfspaces (even with margin) inherits hardness barriers such as PAC learning DNFs.

Finally, \cite{lange2025limitations} showed that membership queries do not provide additional power for testable agnostic learning. In contrast, we show that membership queries can be leveraged in the TDS setting by giving an efficient TDS learner for halfspaces that uses membership queries prior to observing any test examples.

\section{Preliminaries}

\paragraph{Notation} We let $\calX$ be some domain. A concept class $\calC$ over $\calX$ is a set of functions $c:\calX\to \{\pm 1\}$. For a distribution $\calD$ over $\calX$, we denote with $\bar{\calD}$ the corresponding labeled distribution, where the marginal distribution on the labels as well as their coupling with $\calD$ will be clear by the context. In particular, when we say that $\calD$ is labeled by some $f\in\calC$, then the distribution $\bar\calD$ is the distribution over pairs $(\bx,f(\bx))$, where $\bx\sim \calD$. Finally, we will use boldface to denote random variables. 

\subsection{TDS and PQ Learning}

\begin{definition}[TDS Learning \citep{klivans2024testable}]
\label[definition]{definition:tds-learning}
    Let $\calC$ be a concept class over $\calX$. The algorithm $\calA$ is a TDS learner for $\calC$ up to error $\accuracy$ probability of failure $\delta$ if, provided labeled sample access to a training distribution $\dtrainlabeled$ labeled by some unknown $f\in\calC$, and unlabeled sample access to a test distribution $\dtest$, algorithm $\calA$ either outputs $\rej$ or outputs $\acc$ and returns a hypothesis $h:\calX\to \{\pm 1\}$ such that, with probability at least $1-\delta$, the following hold:
    \begin{enumerate}[label=\textnormal{(}\alph*\textnormal{)},leftmargin=6ex]
    \item\label{item:soundness-tds} (\textit{soundness}) Upon acceptance, the test error is bounded as $\Pr_{\bx\sim \dtest}[f(\bx) \neq h(\bx)] \le \accuracy$.
    \item\label{item:completeness-tds} (\textit{completeness}) 
    If $\dtrain =\dtest$, then the algorithm accepts.
    \end{enumerate}
\end{definition}

\begin{definition}[PQ Learning \citep{goldwasser2020beyond}]
\label[definition]{definition:pq-learning}
    Let $\calC$ be a concept class over $\calX$. The algorithm $\calA$ is a PQ learner for $\calC$ up to error $\accuracy$
    and probability of failure $\delta$ if, provided labeled sample access to a training distribution $\dtrainlabeled$ labeled by some unknown $f\in\calC$, and unlabeled sample access to a test distribution $\dtest$, algorithm $\calA$ outputs, with probability at least $1-\delta$, a hypothesis $h:\calX\to \{\pm 1\}$ and a selector $g:\calX\to \{0,1\}$ such that:
    \begin{enumerate}[label=\textnormal{(}\alph*\textnormal{)},leftmargin=6ex]
    \item\label{item:accuracy-pq} (\textit{accuracy}) The test error is bounded as $\Pr_{\bx\sim \dtest}[f(\bx) \neq h(\bx)\text{ and }g(\bx) = 1] \le \accuracy$.
    \item\label{item:rejrate-pq} (\textit{rejection rate}) The probability of rejection is bounded as $\Pr_{\bx\sim \dtrain}[g(\bx) = 0] \le \red{\eps}$.
    \end{enumerate}
\end{definition}

\section{A Reduction from PQ Learning to TDS Learning}

\subsection{Getting a Weak Distinguisher}
\label{sec:weak-distinguisher}
Let $\calA$ be the TDS algorithm and suppose that it draws at most $m$ samples from $\dtrain$ and $\dtest$. We begin by noting that if $\calA$ is likely to reject a set of samples from $\dtest$ then it is clearly able to distinguish between $\dtrain$ and $\dtest$ using $m$ samples. To formalize these statements, we make the following definitions.

\begin{definition}[``Rejecting'' TDS Algorithm]\label[definition]{definition:rejecting-algo}
    Given distributions $\dtrain$ and $\dtest$, we say that a TDS learning algorithm $\calA$ is \emph{rejecting} if $\calA$ outputs $\rej$ with probability at least $0.02$. 
\end{definition}

\begin{definition}[Distinguishing Advantage]
	We say that an algorithm $\calA$ distinguishes $\calD$ and $\calD'$ with $m$ samples and advantage $\gamma$ if 
		\[\left| \Pr_{\ba_1, \dots, \ba_m \sim \calD} [ \calA(\ba_1, \ba_2, \dots, \ba_m) = 1] - \Pr_{\bb_1, \dots, \bb_m \sim \calD'} [ \calA(\bb_1, \bb_2, \dots, \bb_m) = 1]  \right| \geq \gamma \]
\end{definition}

Notably, a rejecting TDS algorithm, $\calA$, has $\Omega(1)$ distinguishing advantage using $m$ samples. Our main goal in this section will be to show, by a hybrid argument, that this implies the existence of an algorithm that distinguishes between $\dtrain$ and $\dtest$ using only a single sample and with distinguishing advantage $\Omega(1/m)$. 

Indeed, we show this precisely with the following lemma. 

\begin{lemma}\label[lemma]{lemma:single-sample-advantage}
    Let $\dtrain$ and $\dtest$ be distributions, $\calA$ be an $(\eps,0.01)$-TDS learner for a class $\calC$ that is rejecting algorithm and uses at most $m$ samples from $\dtest$, and $\delta \in (0,1/3)$, then with probability $1-\delta$, $\getweakdist(\dtrain, \dtest, \calA)$ (cf. \Cref{alg:weak-distinguisher}) outputs an algorithm with distinguishing advantage $\Omega(1/m)$ that only uses a single sample. 
\end{lemma}

\begin{algorithm}[p]
\LinesNumbered
\caption{$\getweakdist(\dtrain,\dtest,\calA,\delta)$}
\label{alg:weak-distinguisher}
\KwIn{Sample access to $\dtrain$ and $\dtest$, a rejecting TDS learning algorithm $\calA$ that uses at most $m$ samples from $\dtest$, and failure probability $\delta\in(0,1)$}
\KwOut{An algorithm $\calA'$ that distinguishes $\dtrain$ and $\dtest$ with advantage $\Omega(1/m)$, or \textsc{FAIL}}
\BlankLine

Let $C \geq C' \geq 1$ be sufficiently large constants and set $\overline{T}\leftarrow \emptyset$, $r \leftarrow C\log(1/\delta)$, and $\ell \leftarrow C m^2 \log(m/\delta)$\;

\BlankLine
\tcc{Find a training set $\overline{T}$ witnessing a noticeable acceptance gap}
\For{$t = 1,2,\dots, C'\log(1/\delta)$}{
\label{line:loop-find-training-ex}
    Sample a set $\ltrain$ from $\ldtrain$\;
    
    Sample sets $\test^{(1)},\dots,\test^{(r)}$ and $\bX_{\textsf{train}}^{(1)},\dots,\bX_{\textsf{train}}^{(r)}$, each of size $m$, from $\dtest$ and $\dtrain$ respectively\;
    
    Compute
    \[
        \wh{p_1} \ :=\ \frac{1}{r}\sum_{i=1}^r\Big(
        \mathbbm{1}\{\calA(\ltrain, \bX_{\textsf{train}}^{(i)})\text{ accepts}\}
        \ -\
        \mathbbm{1}\{\calA(\ltrain, \test^{(i)})\text{ accepts}\}
        \Big)\,;
    \]
    \If{$\wh{p_1} \ge 0.004$}{
        Set $\overline{T} \leftarrow \ltrain$ and \textbf{break}\;
    }
}

\tcc{Search for a single-sample weak distinguisher}
\For{$i = 1,2,\dots, m$}{ \label{line:loop-hybrid}
    \For{$t = 1,2,\dots, C' m\log(1/\delta)$}{
        Sample $\ba_1,\dots,\ba_{m-i} \sim \dtrain$\;
        
        Sample $\bb_1,\dots,\bb_{i-1} \sim \dtest$\;
        
        Define the (randomized) algorithm $\calA'_{\ba,\bb}$: on input $x\in\mathbb{R}^n$, output
        \[
            \calA\big(\overline{T},\{\bb_1,\dots,\bb_{i-1},x,\ba_1,\dots,\ba_{m-i}\}\big)\,;
        \]
        
        Sample $\bc_1,\dots,\bc_{\ell} \sim \dtrain$ and $\bd_1,\dots,\bd_{\ell} \sim \dtest$\;
        
        Compute
        \[
            \wh{p_2}\ :=\ \frac{1}{\ell}\sum_{j=1}^{\ell}\Big(
            \mathbbm{1}\{\calA'_{\ba,\bb}(\bc_j)\text{ accepts}\}
            \ -\
            \mathbbm{1}\{\calA'_{\ba,\bb}(\bd_j)\text{ accepts}\}
            \Big)\,;
        \]
        \If{$\wh{p_2} \ge \frac{1}{5000m}$}{
            \Return $\calA'_{\ba,\bb}$\; 
        }
    }
}
\BlankLine
\Return \textsc{FAIL}\;
\end{algorithm}

We denote by $\calA(\ltrain, \test)$ the result of running $\calA$ on the appropriate number of labeled samples $\ltrain$ drawn from $\ldtrain$ and a test set $\test$ of $m$ samples drawn from $\dtest$. We now start by proving the following lemma.

\begin{lemma}
\label[lemma]{lem:sample-good-train-set}
    \sloppy Consider distributions $\dtrain$ and $\dtest$ and let $\calA$ be a rejecting algorithm drawing $m$ samples from $\dtest$ and $\delta \in (0,1/3)$, then with probability $1-\delta/3$, $\getweakdist(\dtrain, \dtest, \calA, \delta)$ will find a set $\overline{T}$ such that 
 \[\Pr_{\bX \sim (\dtrain)^{\otimes m}} [ \calA(\overline{T}, \bX) \text{ accepts}] - \Pr_{\bX \sim (\dtest)^{\otimes m}} [ \calA(\overline{T}, \bX) \text{ accepts}] \geq 0.003 \]
\end{lemma}

\begin{proof}
We'll assume that
    \[\left| \wh{p_1} - \Pr_{\bX \sim (\dtrain)^{\otimes m}} [\calA(\ltrain,\bX) \text{ accepts}] + \Pr_{\test \sim (\dtest)^{\otimes m}} [\calA(\ltrain,\bX) \text{ accepts}] \right| \leq  0.001 \]
in every iteration of the loop on Line \ref{line:loop-find-training-ex}, which holds with probability at least $1 - \delta/6$ by a Hoeffding bound. Notably, when this event occurs, we will output a set $T$ satisfying the statement of the lemma so long as $T \not= \emptyset$. Thus, it remains to show that the probability of $T \not = \emptyset$ is at most $\delta/6$.

Since $\calA$ is rejecting and rejects samples from $\dtrain$ with probability at most $0.01$, we have that
	\[\Pr_{\ltrain, \bX \sim (\dtrain)^{\otimes m}} [ \calA(\ltrain, \bX) \text{ accepts}] - \Pr_{\ltrain, \bX \sim (\dtest)^{\otimes m}} [ \calA(\ltrain, \bX) \text{ accepts}]  \geq 0.01 \]
So we conclude that 
	\[\E_{\ltrain} \left[ \Pr_{\bX \sim (\dtrain)^{\otimes m}} [ \calA(\ltrain, \bX) \text{ accepts}] - \Pr_{\bX \sim (\dtest)^{\otimes m}} [ \calA(\ltrain, \bX) \text{ accepts}] \right]  \geq 0.01 \]
Since the quantity in the expectation is at most $1$, it follows by reverse Markov that
	\[\Pr_{\ltrain} \left[ \Pr_{\bX \sim (\dtrain)^{\otimes m}} [ \calA(\ltrain, \bX) \text{ accepts}] - \Pr_{\bX \sim (\dtest)^{\otimes m}} [ \calA(\ltrain, \bX) \text{ accepts}] \geq 0.005 \right] \geq 0.005 \]

We now note that if we draw a set $\ltrain$ satisfying the event above, then $\overline{T} \not = \emptyset$. However, since the algorithm repeats $C' \log(1/\delta)$ times, this happens with probability at most $\delta/6$, completing the proof.

\end{proof}

We now are ready to prove \Cref{lemma:single-sample-advantage}.

\begin{proof} of \Cref{lemma:single-sample-advantage}. 
We will assume throughout that 
    \[\left| \wh{p_2} - \Pr_{\bx \sim \dtrain} [\calA'_{\ba,\bb}(\bx) \text{ accepts }] + \Pr_{\bx \sim \dtest} [\calA'_{\ba,\bb}(\bx) \text{ accepts}] \right| \leq \frac{1}{10^5m} \]
in every iteration of the loops on Line \ref{line:loop-hybrid}, which holds by a Hoeffding bound with probability at least $1 - \delta/3$. Moreover, we will assume that the set $\overline{T}$ satisfies
  \[\Pr_{\bX \sim (\dtrain)^{\otimes m}} [ \calA(\overline{T}, \bX) \text{ accepts}] - \Pr_{\bX \sim (\dtest)^{\otimes m}} [ \calA(\overline{T}, \bX) \text{ accepts}] \geq 0.003 \]
	which again happens with probability $1 - \delta/3$ by \Cref{lem:sample-good-train-set}. Assuming both these events hold, it remains to prove that we output FAIL with probability at most $\delta/3$.
    
    For brevity, we will assume that $\calA$ outputs $1$ when it accepts and $0$ when it rejects. To start, we argue that there exists an $i^\star \in [m]$ such that 
	 	\begin{equation}
	 	\label{eq:moo}
	 		\Pr[\calA(\overline{T},\{\bb_1, \dots, \bb_{i^\star-1}, \ba_{i^\star}, \dots \ba_m\}) = 1] - \Pr[\calA(1,\{\bb_1, \dots, \bb_{i^\star}, \ba_{i^\star+1}, \dots \ba_m\}) = 1] \geq \frac{1}{500m} 
	 	\end{equation}
	where $\ba_1, \dots, \ba_m \sim \dtrain$ and $\bb_1, \dots, \bb_m \sim \dtest$. To see this, we apply a hybrid argument. 
 		\begin{align*}
 		0.003 &\leq \Pr[\calA(\overline{T}, \{\ba_1, \dots, \ba_m\}) = 1] - \Pr[\calA(\overline{T}, \{\bb_1, \dots, \bb_m\}) = 1] \\
 		&= \sum_{i = 1}^{m} \Pr[\calA(\overline{T}, \{\bb_1, \dots, \bb_{i-1}, \ba_{i}, \dots \ba_m\}) = 1] - \Pr[\calA(\overline{T}, \{\bb_1, \dots, \bb_{i}, \ba_{i+1}, \dots  \ba_m\}) = 1].
 		\end{align*}
 	An averaging argument then implies that an $i^\star$ as in \Cref{eq:moo} exists.
    
    We now consider the iteration of the loop on Line \ref{line:loop-hybrid} where $i = i^\star$. Note that for random variables $\bc, \ba_1, \dots, \ba_{i-1} \sim \dtrain$ and $\bd, \bb_1, \dots, \bb_{m-i} \sim \dtest$, we have that
	 \begin{align*}
 			&\E_{\ba, \bb, \bc, \bd}[ \calA'_{\ba, \bb}(\bc) - \calA'_{\ba, \bb}(\bd)] \\
 			&= \Pr[\calA(\overline{T}, \{\bb_1, \dots, \bb_{i^\star-1}, \bc, \ba_1 \dots \ba_{m-i}\}) = 1] - \Pr[\calA(\overline{T}, \{\bb_1, \dots, \bb_{i-1}, \bd, \ba_1, \dots  \ba_{m-i}\}) = 1] \\
 			&\geq \frac{1}{500m}
 	\end{align*}
 	 On the other hand, we have
 		\[\E_{\ba, \bb,\bc,\bd}[ \calA'_{\ba, \bb}(\bc) - \calA'_{\ba, \bb}(\bd) ] = \E_{\ba, \bb} \left[ \E_{\bc, \bd} \left[ \calA'_{\ba, \bb}(\bc) - \calA'_{\ba, \bb}(\bd) \right] \right] \]
 	Since $\left| \E_{\bc,\bd} \left[ \calA'_{\ba, \bb}(\bc) - \calA'_{\ba, \bb}(\bd) \right] \right|$ is at most $1$, reverse Markov then yields that with probability $1/1000m$ over $\ba, \bb$, we have that
 	\[ \Pr_{\bx \sim \dtrain} \left[ \calA'_{\ba, \bb}(\bx) = 1 \right] - \Pr_{\bx \sim \dtest} \left[ \calA'_{\ba, \bb}(\bx) = 1 \right] \geq \frac{1}{1000m}. \]
    Note that if we sampled such values $\ba,\bb$ then we would output $\calA'_{\ba,\bb}$ as
        \[\wh{p_2} \geq \frac{1}{1000m} - \frac{1}{10^5m} \geq \frac{1}{2000m}.\]
    by our assumption on the accuracy of $\wh{p_2}$. Thus, the probability of failure is bounded by 
        \[\left(1 - \frac{1}{1000m} \right)^{C' m \log(1/\delta)} \leq \delta/3 \] 
    as desired.	
\end{proof}

\subsection{Boosting the Weak Distinguisher to Obtain a PQ Learner}\label{section:boosting}

Our main result is to provide a PQ learning algorithm that has access to a TDS learning oracle and runs in time polynomial to the oracle's run time. In particular, we prove the following result.

\begin{theorem}[PQ learning via TDS learning]\label[theorem]{theorem:pq-via-tds-main}
    Let $\calA$ be an $(\eps,0.01)$-TDS learner for some class $\calC$ that runs in time $T_\calA(\eps)$ and has total (training plus test) sample complexity $m_\calA(\eps)$. Then, \Cref{algorithm:boosting} is a PQ learner for $\calC$ up to error and rejection rate $O(\eps)$ and probability of failure $\delta$. In particular, \Cref{algorithm:boosting} calls $\calA$ $\poly(m_\calA,1/\eps,\log 1/\delta)$ times, uses $\poly(m_\calA,1/\eps,\log 1/\delta)$ additional samples from the training and test distributions, and runs in time $T_\calA \cdot \poly(m_\calA,1/\eps,\log 1/\delta)$.
\end{theorem}

Before describing the algorithm and the idea of the proof, we make the following observations.

\begin{remark}\label[remark]{remark:output-properties}
The output of $\tdsboost$ (\Cref{algorithm:boosting}) consists of a hypothesis and a selector. Their evaluation requires access to $(\frozen(i,t))_{i,t}$, $(\lbl(i,t))_{i,t}$, and to the functions $(\next_{i,t}^\calA)_{i,t}$. Because evaluating $\next_{i,t}^\calA$ requires oracle access to $\calA$, the representations of the hypothesis and selector must encode sufficient information to invoke this oracle, e.g., by including the code of $\calA$.
Moreover, the hypothesis and the selector are both randomized functions.
\end{remark}

\begin{remark}\label[remark]{remark:beyond-distribution-free}
    In the premise of \Cref{theorem:pq-via-tds-main}, we assume that $\calA$ is a distribution-free TDS learning algorithm, i.e., it satisfies the guarantees of \Cref{definition:tds-learning} for every training distribution $\dtrain$. This assumption can be relaxed. Indeed, \Cref{algorithm:boosting} invokes $\calA$ only on samples drawn from truncated versions of $\dtrain$. Since truncation preserves any property $\calI$ that holds almost surely (e.g., boundedness or the existence of a margin with respect to a target concept $f$), the same argument yields a PQ learner under any distribution $\dtrain$ satisfying $\calI$, provided that $\calA$ is a TDS learner for distributions satisfying $\calI$.
\end{remark}

\begin{algorithm}[p]
\setcounter{AlgoLine}{0}
\LinesNumbered
\caption{$\tdsboost(\calA,\bar{\calD}_1,\Dtest,\eps,\delta)$}\label{algorithm:boosting}
\KwIn{TDS learning oracle $\calA$ for some class $\calC$, sample access to (training) distribution $\bar\calD_1$ realizable by $\calC$, sample access to unlabeled (test) distribution $\Dtest$, and $\eps,\delta\in(0,1)$}
\KwOut{A hypothesis $\hat{h}^{\calA}:\calX\to \{\pm 1 \}$ and a selector $\hat{g}^{\calA}:\calX\to \{0,1\}$}
\BlankLine

Let $m_{\calA}$ be the sample complexity of $\calA$ with parameters $(\eps,0.01)$\;

Let $C\ge 1$ be a sufficiently large constant and let $T=C \,m_\calA^2 \log\frac{1}{\eps}$, and $\delta'= \frac{\delta}{C T^{10}}, r=C\,\log\frac{T}{\eps \delta}, \pmin = \frac{3\eps}{4T(T+1)}, \amin=0.95, \gamma_1 = \frac{\eps}{4T(T+1)}, \gamma_2 = \frac{1}{C\,m_\calA}$\;

\BlankLine
Set $V_{1,1}(\bx) \equiv 1$, $V_{i,t}(\bx) \equiv 0$ for $2\le t\le T$, $1\le i\le t$; \Comment{$V_{i,t}(\bx) = \mathbbm{1}\{\bx\text{ reaches node }(i,t)\}$.}

Set $\next_{i,t}^{\calA}(\bx)\equiv \perp$ for all $1\le i \le t \le T$; \Comment{Routing functions of the branching program}

Set $\lbl(i,T) \leftarrow \mathbbm{1}\{i\ge T/2\}$ for all $1\le i\le T$; \Comment{Assign node to $\calD_b$ (by giving label $b$)}

Set $h_{i,t}(\bx) \equiv 1$ for all $1\le i \le t \le T$; \Comment{Each node corresponds to some hypothesis}

\tcc{Constructing the Branching Program $P=(\next^\calA,\lbl)$}

\For{$t = 1,2,\dots, T$}{
\For{$i = 1, 2, \dots, t$}{
    Compute $(\gamma_1,\delta')$-estimates $\hat p_{i,t}^b$ of $p_{i,t}^b=\Pr_{\bx\sim \calD_b,P}[V_{i,t}(\bx)=1] $ for $b\in\{0,1\}$\;
    
    \If{either of these estimates is less than $\pmin$ and $t<T$}{
        Set $\lbl(i,t) \leftarrow 1-b$, with $\hat p_{i,t}^b < \pmin$; \Comment{If one distribution rarely visits a node, assign the node to the other distribution}
    }
    \Else{
        Set $\calG_{i,t}^b$ to be the conditional distribution $\calD_b \;|\; V_{i,t}(\bx) = 1$ for $b\in\{0,1\}$\;

        Compute $(0.01,\delta')$-estimate $\hat a_{i,t}$ of the probability that algorithm $\calA$ accepts when run on inputs $(\bar\calG_{i,t}^1,\calG_{i,t}^0,\eps,\delta)$\; 

        \If{$\hat a_{i,t} \ge \amin$}{
            Set $\lbl(i,t) \leftarrow 1$\;
            
            Run algorithm $\calA$ on inputs $(\bar\calG_{i,t}^1,\calG_{i,t}^0,\eps,\delta)$ independently $r$ times and let $h_{i,t}$ be the majority function of the hypotheses produced by accepting runs\;
        }
        \ElseIf{$t<T$}{
        Set $\weakdist_{i,t}^{\calA} \leftarrow \getweakdist(\bar\calG_{i,t}^1, \calG_{i,t}^0, \calA, \delta')$ \;
        
        Set $\calG_{i,t}$ to be the mixture $\frac{1}{2}\calG_{i,t}^1+\frac{1}{2}\calG_{i,t}^0$\;
        
        Compute $(\gamma_2,\delta')$-estimate $\hat q_{i,t}$ of probability that $\weakdist_{i,t}^{\calA}(\bx) = 1$ over $\bx\sim \calG_{i,t}$\;
        
        Set $\hat\bw_{i,t}^{\calA}$ to be the randomized function $\hat\bw_{i,t}^{\calA}(\bx) = \textsc{Balance}(\hat q_{i,t}, \weakdist_{i,t}^{\calA}(\bx))$\;
        
        Set $\next_{i,t}^{\calA}$ to be the randomized function $\next_{i,t}^{\calA}(\bx) = (i+\hat\bw_{i,t}^{\calA}(\bx),t+1)$\;
        
        Update the function $V_{i+b,t+1}(\bx)$ to take the following values: $V_{i+b,t+1}(\bx) + V_{i,t}(\bx)\mathbbm{1}\{\hat\bw_{i,t}^{\calA}(\bx)=b\}$ for all $b\in\{0,1\}$\;
        }
    }
}
}

We set the hypothesis $\hat h^{\calA}$ and the selector $\hat g^{\calA}$ to be evaluated as follows:\hspace{7em}
We begin from $N=(1,1)\in [T]^2$ and move to $\next_{N}^{\calA}(\bx)$ recursively until we reach $N\in[T]^2$ with $\next_{N}^{\calA}(\bx) = \perp$.
We set $\hat h^{\calA}(\bx)=h_N(\bx)$ and $\hat g^{\calA}(\bx) = \lbl_{N}$\;
\end{algorithm}

\paragraph{Algorithm Idea} $\tdsboost$ (\Cref{algorithm:boosting}) is an extension of the Martingale boosting algorithm of \cite{long2005martingale}, which was proposed in the context of boosting a weak PAC learner to obtain a strong PAC learning algorithm with favorable noise tolerance properties compared to prior work. Here, we are building on these ideas in the context of learning with distribution shift, by using them to boost the distinguishing advantage provided by a rejecting TDS learner according to \Cref{lemma:single-sample-advantage}. 

In particular, we consider the following contrived learning problem. Suppose that we wish to learn a function $g: \calX\times \{0,1\} \to \{0,1\}$, where $g(x,b) = b$, under a distribution $\calG$. To draw $(\bx,\bb)$ from $\calG$, with probability $\frac{1}{2}$, we draw $\bx \sim \Dtrain$ and set $\bb=1$. Otherwise, we draw $\bx \sim \Dtest$ and set $\bb=0$. Due to \Cref{lemma:single-sample-advantage}, as long as the TDS learner $\calA$ rejects with probability at least $\Omega(1)$ when run on the distributions $\bar\calD_1, \Dtest$, we have that \Cref{alg:weak-distinguisher} is a weak learner for the above learning problem, i.e., outputs a hypothesis $\weakdist^\calA:\calX\to \{0,1\}$ such that:
\[
    \Pr_{(\bx,\bb)\sim \calG}\Bigr[\weakdist^\calA(\bx) = g(\bx,\bb) \Bigr] \ge \frac{1}{2} + \gamma\,,
\]
where $\gamma = \Omega(1/m_\calA)$ and $m_\calA$ is the total sample complexity of $\calA$. For the remainder of this subsection, we will use the notation $\Dtrain$ to denote the training distribution $\dtrain$, and $\Dtest$ to denote the test distribution $\dtest$.\footnote{This notation is convenient because we can use $\calD_{b}$ for $b\in\{0,1\}$ to refer to either the training distribution ($b=1$, and also $g(\bx,b)=1$) or the test distribution ($b=0$, and also $g(\bx,b)=0$).}

Note that since $\calA$ is a distribution-free TDS learner, at a first glance, one might think that \Cref{alg:weak-distinguisher} is a distribution-free weak learner for the problem of learning $g$. If that were indeed the case, then we could apply Martingale boosting in order to obtain a strong learner for $g$ whose output $\hat g$ only depends on $\bx$ and not $b$.\footnote{The same would be true if we applied any boosting algorithm that outputs a hypothesis $\hat g$ with the following property: If the outputs of the weak learners only depend on $\bx$ and not $b$, then $\hat g$ is a junta on $\bx$ as well.} This, in turn, would give a PQ learner, because $\hat g$ would strongly distinguish between the training and test distributions and we would, therefore, have the following high-probability bounds on the error and rejection rates:
\begin{align*}
    \Pr_{\bx\sim \Dtest}[1\neq f(\bx), \hat g(\bx) = 1] &\le \Pr_{\bx\sim \Dtest}[ \hat g(\bx) \neq 0] \le \eps \tag{Here we used the trivial hypothesis $\hat h \equiv 1$} \\
    \Pr_{\bx\sim \Dtrain}[\hat g(\bx) = 0] &= \Pr_{\bx\sim \Dtrain}[\hat g(\bx) \neq 1]\le \eps
\end{align*}

However, the above attempt fails, because \Cref{alg:weak-distinguisher} is not really distribution-free. In particular, \Cref{alg:weak-distinguisher} is guaranteed to work only when the distributions $\Dtrain$ and $\Dtest$ are such that the algorithm $\calA$ rejects with non-trivial probability when run on samples from these distributions (see \Cref{definition:rejecting-algo}). Fortunately, if $\calA$ rejects with very small probability, then we may use it to obtain a hypothesis with low test error. Moreover, we can incorporate this case in the Martingale boosting algorithm without introducing any technical complications, thereby obtaining the desired result.

\paragraph{Algorithm Description} The algorithm $\tdsboost$ creates a branching program and uses it to partition the domain $\calX$ into a collection of regions for each of which it is either safe to abstain (selector takes the value $0$) or to predict (selector takes the value $1$).\footnote{The branching program, and hence the corresponding partition of the domain $\calX$, is randomized, as its routing functions $\next_{i,t}^\calA$ are randomized.} 

The branching program consists of nodes of the form $(i,t)\in [T]^2$ with $i\le t$, routing functions for each node $\next_{i,t}^\calA:\calX\to \{(i,t+1),(i+1,t+1),\perp\}$, labels $\lbl_{i,t}\in \{0,1\}$ for leaf nodes (i.e., nodes with $\next_{i,t}^\calA = \perp$), and auxiliary functions $h_{i,t}:\calX\to \{\pm 1\}$ for leaf nodes. 

A point $x \in \calX$ starts from node $N=(1,1)$ and is routed through the branching program by applying the function $\next_N^\calA(x)$ recursively until it reaches a leaf node $N$ with $\next_N^{\calA}(x)=\perp$. The aforementioned partition contains one region for each leaf node $N$ consisting of the points $x$ that are routed to $N$. The label $\lbl_N$ of a leaf node corresponds to the suggested value of the selector, and the hypothesis $h_N$ corresponds to the suggested classifier for the elements of the region corresponding to node $N$. 

The algorithm creates the branching program using an iterative top-down approach. It considers for each node $(i,t)$ the distributions $\calG_{i,t}^1$ and $\calG_{i,t}^0$ which correspond to $\Dtrain$ and $\Dtest$, respectively, conditioned on visiting node $(i,t)$ through the branching program constructed so far. 

If sampling from either of these distributions (through rejection sampling) is too costly, then at least one of $\Dtrain, \Dtest$ rarely visits the node $(i,t)$ and we make $(i,t)$ a leaf node, with label opposite to the distribution that rarely visits it. Recall that the labels of the leaf nodes determine the values of the selector. This step is safe to do because it is always safe to abstain on regions of low probability under $\Dtrain$, as well as give predictions on regions of low probability under $\Dtest$ (see \Cref{definition:pq-learning}).

The aforementioned conditioning defines a new TDS learning problem. If the given TDS learner $\calA$ is very likely to accept, then we use it to obtain a hypothesis $h_{i,t}$ with high accuracy on $\calG_{i,t}^0$ and set its label to $1$. 
Otherwise, the algorithm $\calA$ rejects with non-trivial probability and, by \Cref{lemma:single-sample-advantage}, we may use \Cref{alg:weak-distinguisher} to obtain a weak distinguisher $\weakdist_{i,t}^\calA$ for the distributions $\calG_{i,t}^1$ and $\calG_{i,t}^0$. Having $\weakdist_{i,t}^\calA$ at hand, we are able to construct the routing function $\next_{i,j}^\calA$ by using the rebalancing trick of \cite{long2005martingale}, which is given by \Cref{algorithm:balancing}.

\begin{algorithm}[h]
\setcounter{AlgoLine}{0}
\LinesNumbered
\caption{$\textsc{Balance}(q,w)$}\label{algorithm:balancing}
\KwIn{$q\in[0,1],\; w\in\{0,1\}$}
\KwOut{A random variable $\hat w$ over $\{0,1\}$}
\BlankLine

Let $b = \argmin_{b'\in\{0,1\}} |b' - q|$ and $\xi$ be Bernoulli random variable that is $1$ with probability $\frac{1}{1 + 2|b-q|}$\;

Set $\hat w \leftarrow w\,\xi + (1-b)(1-\xi)$\;
\end{algorithm}

\paragraph{Comparison to Martingale Boosting} Similarly to standard Martingale boosting, the routing functions are chosen so that samples from the distribution $\Dtrain$ are biased to move towards nodes $(i,t)$ where $i$ is large, and samples from the distribution $\Dtest$ are biased towards nodes with $i$ small. This is accomplished by using the weak distinguishers provided by \Cref{alg:weak-distinguisher} to construct the routing functions. In particular, we inherit the following result from \cite{long2005martingale} which essentially shows that nodes of the form $(i,T)$ with $i<T/2$ are rarely visited by samples of $\Dtrain$ and nodes $(j,t)$ with $i\ge T/2$ are rarely visited by samples of $\Dtest$.

\begin{lemma}[Implicit in Theorem 3 from \cite{long2005martingale}]
\label[lemma]{lemma:martingale-boost}
    Let $\calG$ be the following distribution over $\calX\times\{0,1\}$. To draw $(\bx,\bb)$ from $\calG$, with probability $\frac{1}{2}$, we draw $\bz \sim \Dtrain$ and set $\bx = \bz$ and $\bb = 1$. Otherwise, we draw $\bz \sim \Dtest$ and set $\bx = \bz$ and $\bb = 0$. Suppose that we obtain the randomized branching program $P$ by running $\tdsboost(\calA,\bar\calD_1, \Dtest, \eps,\delta)$ where $\calA$ is an $(\eps,\delta)$-TDS learner. Then we have:
    \[
        \Pr_{(\bx,\bb)\sim \calG, P}\Bigr[\bx \text{ reaches node }(i,T)\text{ in }P \text{ and } \lbl(i,T)\neq \bb\Bigr] \le \frac{\eps}{2}
    \]
\end{lemma}

The main difference between \Cref{algorithm:boosting} and the Martingale boosting algorithm of \cite{long2005martingale} is that our algorithm creates some additional leaf nodes at levels $t<T$. However, this does not change the analysis of \Cref{lemma:martingale-boost}, as the routing functions of the internal nodes are created based on valid weak distinguishers given by \Cref{alg:weak-distinguisher}.

\paragraph{Sketch of the Proof} As we have mentioned, \Cref{algorithm:boosting} creates a branching program $P$ that partitions the domain $\calX$ into regions $\{\calR_N\}_{N\in \leaves(P)}$ corresponding to the leaves $\leaves(P)$ of $P$. For each such region $\calR_N$, we either abstain if $\lbl(N)=0$ or predict according to hypothesis $h_N$ if $\lbl(N)=1$. There are three types of leaf nodes:
\begin{enumerate}
    \item The first type $\calT_1\subseteq\leaves(P)$ includes nodes $N$ at level $t<T$ that are rarely visited by $\calD_b$ for some $b\in\{0,1\}$ and $\lbl(N) = 1-b$. If $b=1$, then the training distribution assigns negligible mass to $\calR_N$ and hence abstaining on this region does not increase the rejection rate significantly. If $b=0$, then the test distribution assigns negligible mass to $\calR_N$ and hence even if we make wrong predictions, the error rate is not increased significantly.
    \item The second type $\calT_2\subseteq\leaves(P)$ includes the nodes at level $t=T$. Due to \Cref{lemma:martingale-boost}, these nodes have a similar property, i.e., $\Pr_{\bx\sim\calD_b,P}[\bx \text{ reaches }\calT_2\text{ at node with label } 1-b]\le O(\eps)$.
    \item Finally, the third type $\calT_3\subseteq\leaves(P)$ are those nodes $N$ where the TDS learner accepts with probability at least $0.9$, which enables us to learn a hypothesis with low test error conditioned on reaching node $N$, by amplifying the acceptance probability through the majority vote of several repetitions of the TDS learner (based on Proposition C.1 of \cite{klivans2024testable}).
\end{enumerate}

Following this sketch, we now provide a formal proof for \Cref{theorem:pq-via-tds-main}.

\begin{proof} of \Cref{theorem:pq-via-tds-main}. 
We assume that we are in the event that happens with probability at least $1-\delta$ such that the following conditions hold:
\begin{itemize}
    \item\label{item:first-type-leaves} Our estimates $\hat p_{i,t}^b, \hat a_{i,t}, \hat q_{i,t}$ are within $\gamma_1,0.01,\gamma_2$ of their expected values respectively.
    \item\label{item:second-type-leaves} Any time we call \Cref{alg:weak-distinguisher}, it outputs a valid weak distinguisher.
    \item\label{item:third-type-leaves} The majority boosted TDS learner never fails when called, i.e., it outputs $\acc$ and gives a hypothesis $h_{i,t}$ with error at most $4 \eps$ on $\calG^0_{i,t}$.
\end{itemize}
For the first property, we may use standard Hoeffding bounds to show that our estimates --- which we obtain through averaging --- are concentrated around their expectations. To show that the last property holds with high probability, we combine Proposition C.1 of \cite{klivans2024testable} with the observation that majority boosting amplifies the probability of acceptance as long as the probability of acceptance is larger than $1/2$.

We will first bound the rejection rate on the training distribution. We have the following:
\begin{align*}
    \Pr_{\bx\sim \Dtrain, P}[\hat g(\bx) = 0] &\le \Pr_{\bx\sim \Dtrain, P}[\exists N\in \leaves(P) : \bx \text{ reaches }N \wedge \lbl(N) = 0] \\
    &\le (\pmin+\gamma_1) T^2 + \eps \tag{first and second type leaves}\\
    &\le O(\eps)
\end{align*} 
To obtain the above bound we used the property that each leaf $N\in\calT_1$ with $\lbl(i,t) = 0$ satisfies $\Pr_{\bx\sim \Dtrain,P}[\bx \text{ reaches }N] \le \pmin+\gamma_1$, as well as \Cref{lemma:martingale-boost} for leaves of the form $(i,T)$.

For the error rate we obtain similar bounds for the contributions of the first and second types of leaf nodes, but we also need to account for the leaves of third type. Fortunately, any such node $N$ corresponds to hypothesis with low error on the test distribution conditioned on reaching node $N$. We have the following:
\begin{align*}
    \Pr_{\bx\sim \Dtest, P}&[\hat h(\bx)\neq f(\bx) \wedge \hat g(\bx) = 1] \\
    &\le \Pr_{\bx\sim \Dtest, P}[\exists N\in \leaves(P) : \bx \text{ reaches }N \wedge \lbl(N) = 1 \wedge h_N(\bx)\neq f(\bx)] \\
    &\le (\pmin+\gamma_1) T^2 + \eps + 4\eps \tag{first, second and third type leaves}\\
    &\le O(\eps)
\end{align*}

Finally, the bound on the run time follows from the fact that the boosting procedure runs in time $\poly(1/\gamma, \frac{1}{\eps}, \log(1/\delta), m_\calA) \cdot T(\calC, \eps)$, where $\gamma = \Omega ( {1}/{m_\calA} )$ is the advantage of the weak learner from \Cref{lemma:single-sample-advantage}. 
\end{proof}

\section{Learning Halfspaces with Distribution Shift with Queries}\label{section:halfspaces}
We note that a simple corollary of \Cref{theorem:pq-via-tds-main} is that it is unlikely that there is a polynomial time algorithm in $n$ and $\eps^{-1}$ to TDS learn halfspaces in the distribution-free setting. Indeed, combining \Cref{theorem:pq-via-tds-main} with \cite{kalai2021efficient}, which proves that PQ learning and reliable learning are equivalent, implies that TDS learning of halfspaces to accuracy $\eps$ is at least as hard as learning an intersection of $\Omega(1/\eps)$ halfspaces in the distribution-free PAC model.
Notably, whether there exists a polynomial-time algorithm for learning even an intersection of even two halfspaces is a long standing open question. 
Moreover, under suitable cryptographic assumptions \citep{klivans2009cryptographic, tiegel2024improved}, learning an intersection of polynomially many halfspaces requires sub-exponential time.
Even for halfspaces with a margin, it is unlikely there is a polynomial time algorithm for TDS learning. In particular, if we consider learning conjunctions over $\zo^n$, then every point has an inverse polynomial margin. By \Cref{theorem:pq-via-tds-main} and \Cref{remark:beyond-distribution-free}, it then follows that we can reduce TDS learning of conjunctions over $\{0,1\}^n$ to PQ learning of conjunctions. In turn, \cite{kalai2021efficient} showed PQ learning conjunctions is as hard as PAC learning DNFs, another long standing open problem in learning theory believed to not have polynomial time algorithms.

All that said, in this section, we will show that if we allow ourselves membership queries, then we can PQ and thus TDS learn halfspaces in polynomial time.
In particular, we define a PQ learner with membership queries as follows:

\begin{definition}[PQ Learning with Membership Queries]
    Let $\calC$ be a concept class over $\calX$. An PQ learning algorithm with membership queries, $\calA$, for the class $\calC$ up to error $\accuracy$ and probability of failure $\delta$ operates as follows. First, the algorithm is given labeled sample access to a training distribution $\ldtrain$ labeled by some unknown $f\in\calC$ and can query arbitrary points $x \in \calX$ to receive the value $f(x)$. Afterward, the algorithm is given unlabeled samples access to $\dtest$ and must (without making any further queries) output, with probability at least $1-\delta$, a hypothesis $h:\calX\to \{\pm 1\}$ and a selector $g:\calX\to \{0,1\}$ such that:
    \begin{enumerate}[label=\textnormal{(}\alph*\textnormal{)},leftmargin=6ex]
    \item (\textit{accuracy}) The test error is bounded as $\Pr_{\bx\sim \dtest}[f(\bx) \neq h(\bx)\text{ and }g(\bx) = 1] \le \accuracy$.
    \item (\textit{rejection rate}) The probability of rejection is bounded as $\Pr_{\bx\sim \dtrain}[g(\bx) = 0] \le \eps$.
    \end{enumerate}
\end{definition}

In this setting, we will then prove
\begin{theorem}
\label[theorem]{thm:MQ-halfspaces}
    There exists a PQ learning algorithm for halfspaces over $\mathbb{R}^n$ with membership queries that runs in time $\poly(n, \frac{1}{\eps}, \log(1/\delta))$.
\end{theorem}

\begin{remark}
Our upper bound for learning halfspaces with membership queries holds even under the stronger \emph{reliable and probably useful} (RPU) model. Unlike PQ learning, an RPU learner is only given (labeled) training examples, and outputs a selector $g$ and a classifier $h$ such that: whenever a point $x$ is selected ($g(x) = 1$), the prediction is always correct, i.e., $h(x) = f^*(x)$.

At first glance, this may appear to contradict prior exponential lower bounds for RPU learning with membership queries \cite{rpu}. However, those results assume that queries are restricted to a bounded-radius ball, whereas our algorithm allows arbitrary queries in $\mathbb{R}^n$.
\end{remark}

The rest of this section presents the proof of \Cref{thm:MQ-halfspaces}. To prove the statement, we first show that we can essentially reduce to homogeneous halfspaces. Afterwards, we give a PQ learning algorithm for homogenous halfspaces with a margin. Finally, we combine our algorithm for PQ learning high-margin halfspaces with the Forster transform to prove \Cref{thm:forster}.

\subsection{A Reduction to Homogeneous Halfspaces}
\label{sec:apx-halfspaces-reduction}
In this section, we show that PQ learning halfspaces with a margin reduces to a restricted class of query algorithms for homogeneous halfspaces. In particular, we define

\begin{definition}[Respectful Algorithm]
    Suppose that we are given a query algorithm $\calA$ to PQ learn a halfspace $f(x) := \sgn(w \cdot x)$ in $\mathbb{R}^{n}$. We say that $\calA$ is respectful if whenever $\dtrain$ is supported on $\{(x,1): x \in \mathbb{R}^{n-1}\}$, we have that $\calA$ almost surely never queries a point of the form $(x,0)$ for some $x \in \R^{n-1}$.
\end{definition}

We then have that
\begin{lemma}
\label[lemma]{lem:respectful}
Suppose that there exists a respectful PQ learning algorithm for homogeneous halfspaces over $\mathbb{R}^n$ with membership queries that runs in time $\poly(n, \frac{1}{\eps}, \log(1/\delta))$, then \Cref{thm:MQ-halfspaces} holds.\footnote{We quickly remark that while we will work in the real model of computation for simplicity, we note that if there is a respectful homogenous halfspace PQ learning algorithm that only queries points with polynomial bit complexity then the algorithm obtained for arbitrary halfspaces via the reduction will also only query points of polynomial bit complexity.}
\end{lemma}

\begin{proof}
    Let $\calA$ denote our respectful algorithm for homogeneous halfspaces and suppose that we are tasked with learning an unknown halfspace $f(x) = \sgn(w \cdot x - \theta)$ with distributions $\dtrain$ and $\dtest$. Let $\Phi$ be the map that sends a point $x \in \R^n$ to $(1,x) \in \R^{n+1}$ and denote by $\Phi(\dtrain)$ and $\Phi(\dtest)$ denote the distributions resulting from drawing from $\dtrain$ and $\dtest$ respectively and then applying $\Phi$.
    Note that after applying this transformation, the target halfspace can be written as $\sgn((w, -\theta) \cdot \Phi(x))$ and thus correspond to a homogeneous halfspace after our transformation.

    We will now run $\calA$ on $\Phi(\dtrain)$ and $\Phi(\dtest)$. To do so, we must argue that we can simulate membership queries in this new space. Indeed, suppose that $\calA$ queries a point $(x,c)$ for some $c \in \R$ and $x \in \R^n$. If $c = 0$, we simply output FAIL. Otherwise, we can output $f(x/c) \cdot \sgn(c)$. Since
        \[f(x/c) \cdot \sgn(c) = \sign(w \cdot (x/c) - \theta) \cdot \sgn(c) = \sgn(w \cdot x - c \theta) = \sign((w, -\theta) \cdot (x,c)),\]
    this correctly answers the membership query.

    We now note that we fail with measure zero as $\calA$ is respectful, thus this process almost surely outputs a hypothesis $h$ and selector $g$. Using $h \circ \Phi$ and $g \circ \Phi$, then yields the hypothesis and selector for PQ learning $f$.
\end{proof}

\subsection{PQ Learning Homogeneous Halfspaces with a Margin}
\label{sec:apx-halfspaces-margin}

We now turn to showing how to PQ learn a homogenous halfspaces, i.e. those of the form $f(x) = \sgn(w \cdot x)$, with a margin (cf. \Cref{def:margin}). As is standard in this setting, we will assume without loss of generality that $\dtrain$ and $\dtest$ are supported on the unit sphere $\mathbb{S}^{n-1}$ and $\|w\|_2 = 1$.

\begin{algorithm}[h]
\setcounter{AlgoLine}{0}
\LinesNumbered
\caption{$\learnmargin(f,\gamma,\delta)$}\label{alg:pq-large-margin}
\KwIn{Query access to $f$, a margin parameter $\gamma \in (0,1)$, and failure parameter $\delta \in (0,1/3)$}
\KwOut{A hypothesis $h:\calX\to\{\pm 1\}$ and a selector $g:\calX\to\{0,1\}$}
\BlankLine

Let $\ell \gets \frac{2000 n}{\gamma^2}\log\!\left(\frac{n}{\delta}\right)$\;

Sample $\bx^{(1)},\dots,\bx^{(\ell)} \sim \calN\!\left(0,\frac{\pi}{2}I_n\right)$\;

Compute $\wh{w} \gets \frac{1}{\ell}\sum_{i=1}^{\ell} f(\bx^{(i)})\,\bx^{(i)}$\;

Set $h(x) \gets \sgn(\wh{w}\cdot x)$ and $g(x) \gets \mathbb{I}\!\left(\left|\wh{w}\cdot x\right|\ge \frac{2\gamma}{3}\right)$\;

\Return $(h,g)$\;
\end{algorithm}

Our main lemma in this section will then be

\begin{lemma}
\label[lemma]{lem:learn-margin-guarantee}
Let $f$ denote an unknown halfspace, them with probability at least $1-\delta$, we have that $\learnmargin(f, \gamma, \delta)$ outputs a hypothesis $h$ such that for all for all $x \in \mathbb{S}^{n-1}$ $(i)$ if $x$ has margin at least $\gamma$ with respect to $f$, then $g(x) = 1$ and $(ii)$ if $g(x) = 1$ then $h(x) = f(x)$ and $(iii)$ $\|\wh{w}\| \geq \frac{2}{3}$
\end{lemma}

\begin{proof}
	We will show that the lemma follows from the following claim
    \begin{claim}
        \[\Pr \left[ \| \wh{w} - w \|_2 \geq \gamma/3 \right] \leq \delta \]    
    \end{claim}

    \begin{proof}
        By rotational symmetric, we can let $w = e_1$ without loss of generality. We then note that
        \[\E[f(\bx^{(i)}) \cdot \bx^{(i)}] = e_1 \]
    Indeed, by symmetry the expectation of of the $j$th coordinate is $0$ for all $j \not = 1$. We can then compute the first expectation of the first coordinate as 
        \[\E_{\bz \sim \calN(0,\frac{\pi}{2})} [|\bz|] = \sqrt{ \frac{2}{\pi}} \cdot \sqrt{ \frac{\pi}{2}} = 1 \]
    We now note that for all $j \in [n]$, $\left( f(\bx^{(i)}) \cdot \bx^{(i)} - e_i \right)_j$ is sub-gaussian with variance proxy at most $10$. Thus it follows that for any $j$, we have that 
        \[ \Pr[|\wh{w}_j - (e_i)_j| \geq \gamma/3\sqrt{n} ] \leq \exp(- \ell \gamma^2/2000n) \leq \frac{\delta}{n}  \]
    Thus, applying a union bound then yields that for all $j \in [n]$
        \[|\wh{w}_j - (e_i)_j| \leq \gamma/3\sqrt{n}\]
    with probability $1- \delta$. Notably, this implies that 
    \[\|\wh{w} - e_i\|_2 \leq \gamma/3\]
    with probability at least $1 - \delta$ as desired.
    \end{proof}

    We will now show that conditioned on $\|\wh{w} - w\|_2 \leq \gamma/3$, properties $(i)$, $(ii)$, and $(iii)$ hold. Indeed, the triangle inequality immediately implies $\wh{w} \geq 1 - \gamma/3 \geq \frac{2}{3}$. Next, for property $(i)$, note that 
        \[\left| w \cdot x - \wh{w} \cdot x \right| \leq \|w-\wh{w}\|_2 \leq \gamma/3 \]
    Thus it follows that if $|w \cdot x| \geq \gamma$ then $|\wh{w} \cdot x| \geq 2 \gamma / 3$ and thus $g(x) = 1$. Similarly, we can observe that if $|\wh{w} \cdot x| \geq 2 \gamma / 3$, we have that $\sgn(w \cdot x) = \sgn( \wh{w} \cdot x)$, yielding property $(ii)$.
\end{proof}

\subsection{Proof of \Cref{thm:MQ-halfspaces}}
\label{sec:apx-halfspaces-forster}
In this section, we give an algorithm to prove \Cref{thm:MQ-halfspaces}.  To do so, our algorithm \learnh{} will apply a Forster transform to ensure that a non-trivial fraction of points have a large margin. We can then run \learnmargin{} to classify these points. Afterwards, we can repeat these steps for the points in $\train$ that were not labeled.

\begin{algorithm}[h]
\setcounter{AlgoLine}{0}
\LinesNumbered
\caption{$\learnh(\train,\eps,\delta)$}\label{alg:pq-halfspace}
\KwIn{A set examples $\train$ drawn i.i.d. from $\dtrain$, a rejection parameter $\eps$, and a failure probability $\delta$}
\KwOut{A hypothesis $h:\calX\to\{\pm1\}$ and a selector $g:\calX\to\{0,1\}$}
\BlankLine

$R_0 \gets \train$, $i \gets 0$\;

Set $g_0(x)\equiv 0$ and $h_0(x)\equiv 1$\;
\BlankLine

\While{$\frac{\lvert R_i\rvert}{\lvert \train\rvert} \ge \frac{1}{2}\eps$ \textbf{and} $i < 4n\log\frac{2}{\eps}$}{
    $(A_i,V_i) \gets \forster\!\left(R_i,\frac{\delta}{24n\log(2/\eps)}\right)$\;
    
    \If{$R_i \cap V_i = \emptyset$}{
    \label{line:empty-subspace}
            \textbf{break}\;
        }
        
    $(\ell_i,s_i) \gets \learnmargin\!\left(f(P_{V_i}^T A_i^{-1}x),\frac{1}{2\sqrt{n}},\frac{\delta}{24n\log(2/\eps)}\right)$\;

    $R_{i+1} \gets R_i \setminus
    \left\{x\in R_i\cap V_i:
    s_i\!\left(\frac{A_i P_{V_i}x}{\lVert A_i P_{V_i}x\rVert}\right)=1
    \right\}$\;

    Define
    \[
      g_{i+1}(x) \gets g_i(x)\ \lor\
      \Big(\mathbbm{1}\{x\in V_i\}\ \land\ s_i\!\left(\frac{A_i P_{V_i}x}{\lVert A_i P_{V_i}x\rVert}\right)\Big)
      \quad\text{for all }x
    \]

    Define
    \[
      h_{i+1}(x) \gets
      \begin{cases}
        h_i(x) & \text{if } g_i(x)=1,\\[0.25em]
        \ell_i\!\left(\frac{A_i P_{V_i}x}{\lVert A_i P_{V_i}x\rVert}\right) & \text{ otherwise}
      \end{cases}
    \]

    $i \gets i+1$\;
}

\BlankLine
\Return{$(h_i,g_i)$}\;
\end{algorithm}

To analyze \learnh{}, we start by showing that we typically terminate when we have labeled all but an $\eps/2$ fraction of the data set.

\begin{lemma}
\label[lemma]{lem:train-labeled}
	With probability $1 - \delta/3$, when $\learnh(\train, \eps, \delta)$ terminates we have that $\frac{|R_i|}{|\train|} \leq \frac{\eps}{2}$.
\end{lemma}

\begin{proof}
	We'll assume that no call to \forster{} or \learnmargin{} fails, which by a union bound happens with probability at least $1 - \delta/3$. Notably, in this case by \Cref{cor:forster}, we will never have $V_i \cap R_i = \emptyset$ and thus never execute Line \ref{line:empty-subspace}.
    
    Now applying \Cref{lem:forster-margin}, it then follows that  
		\[\Pr_{\bx \sim R_i \cap V_i} \left[ \left| \frac{\wh{w}}{\|\wh{w}\|} \cdot \frac{A_i P_{V_i} \bx}{\|A_i P_{V_i} \bx\|} \right| \geq \frac{1}{2 \sqrt{n}} \right] \geq \frac{1}{4\dim(V_i)} \]
	Since $\|\wh{w}\| \geq \frac{2}{3}$ by the promise of \Cref{lem:learn-margin-guarantee}, it follows that
		\[\Pr_{\bx \sim R_i \cap V_i} \left[ \left| \wh{w} \cdot \frac{A_i P_{V_i} x}{\|A_i P_{V_i} x\|} \right| \geq \frac{1}{3 \sqrt{n}} \right] \geq \frac{1}{4\dim(V_i)} \]
	Note that the points in $R_i \cap V_i$ satisfying the above event also satisfy $\mathbb{I}(x \in V_i) \land s_i \left( \frac{A_i P_{V_i} x}{\|A_i P_{V_i} x\|} \right)$. So, it follows that
		\[\frac{|R_{i+1}|}{|R_i|} \leq 1 - \frac{1}{4 \dim(V_i)} \cdot \frac{|R_i \cap V_i|}{|R_i|} \leq 1 - \frac{1}{4 \dim(V_i)} \cdot \frac{\dim(V_i)}{n} \frac{|R_i|}{|R_i|} = 1 - \frac{1}{4n}  \]
	The result now follows by noting that
		\[\frac{|R_i|}{|\train|} \leq \left(1 - \frac{1}{4n} \right)^{4n \log(2/\eps)} \leq \frac{\eps}{2} \]
\end{proof}

\begin{lemma}
\label[lemma]{lem:rej-prob}
Suppose that $\train$ is a set of size $\Omega \left (\frac{Tn^2 \log(nT) \log(1/\eps \delta)}{\eps^2} \right)$ samples from $\dtrain$, then with probability at least $1 - \delta/2$ the hypothesis $h$ outputted by $\learnh(\train, \eps, \delta)$ satisfies
\[ \Pr_{\bx \sim \dtrain}[g(\bx) = 0] \leq \eps\]	
\end{lemma}

\begin{proof}
Let $T$ denote the number of iterations of the while loop. We will argue that $g_1, \dots, g_T$ come from a class of VC dimension $O(T n^2 \log(nT))$. Towards this, we first show that the function $\mathbb{I}(x \in V_i) \land s_i \left( \frac{A_ix}{\|A_ix\|} \right)$ has VC dimension at most $O(n^2)$. Indeed, the indicator of a subspace has VC dimension at most $O(n)$. Using the definition of $s_i$ as appearing in \learnmargin, we have that
		\[s_i \left( \frac{A_i P_{V_i} x}{\|A_i P_{V_i} x\|} \right) = \mathbb{I} \left( \wh{w} \cdot \frac{A_i P_{V_i} x}{\|A_i P_{V_i} x\|} \geq \frac{1}{3 \sqrt{n}} \right) \]
	Notably, this can be rewritten as a conjunction of a degree $2$ PTF and a halfspace, namely
		\[\sgn \left( (\wh{w} \cdot A_iP_{V_i}x)^2 - \frac{1}{9n} \cdot \|A_iP_{V_i}x\|^2\right) \land \sgn \left( \wh{w} \cdot A_iP_{V_i} x \right).\]
	As such, it follows that $\mathbb{I}(x \in V_i) \land s_i \left( \frac{A_iP_{V_i}x}{\|A_iP_{V_i}x\|} \right)$ is the conjunction of a subspace, halfspace, and degree $2$ PTF. Using standard results, it thus follows that it has VC dimension $O(n^2)$ as desired. Since each $g_i$ is a decision list of length at most $T$ over features of the form $\mathbb{I}(x \in V_i) \land s_i \left( \frac{A_ix}{\|A_ix\|} \right)$, it follows from standard techniques that the functions $g_i$ for $i < T$ have VC dimension at most $O(Tn^2\log(nT))$ (see e.g. Theorem 3.6 of \cite{kearns1994introduction}). By \Cref{cor:VC-sampling}, it then follows that with probability $1 - \delta/6$
		\[\Pr_{\bx \sim \dtrain}[g_T(\bx) = 0] = \frac{|R_T|}{|\train|} + \frac{\eps}{2}.\]
    Thus, if $\Pr_{\bx \sim \dtrain}[g(\bx) = 0] > \eps$ either $\frac{|R_T|}{|\train|} > \frac{\eps}{2}$, which happens with probability at most $\delta/3$ by \Cref{lem:train-labeled}, or the VC inequality fails. So, by a union bound, the lemma statement holds with probability $1 - \delta/2$, as desired.
\end{proof}

With this we can complete the proof of \Cref{thm:MQ-halfspaces}. 

\begin{proof} of \Cref{thm:MQ-halfspaces}. 
We start by noting that the \learnh{} successfully PQ learns a homogeneous halfspace assuming $|\train|$ is sufficiently large so as to satisfy the hypothesis of \Cref{lem:rej-prob}. Indeed, with probability $1 - \delta/2$, the rejection probability of $g$ under $\dtrain$ is at most $\eps$ by \Cref{lem:rej-prob}. Moreover, so long as no call to \learnmargin{} fails, which happens with probability at least $1 - \delta/6$ by a union bound, we have that our algorithm has zero error on points with $g_i(x) = 1$ by \Cref{lem:learn-margin-guarantee}. In particular, $\learnmargin{}$ outputs a hypothesis $\ell_j$ such that $\ell_j(x) = \sgn(P_{V_j} w \cdot A_j^{-1} x)$ for all $x \in \mathbb{S}^{n-1} \cap V_j$ such that $s_j(x) = 1$. As such we conclude that $\ell_j \left( \frac{A_j P_{V_j} x}{\|A_j P_{V_j} x\|} \right) = \sgn(w \cdot x)$ for all $x \in V_j$ satisfying $s_j \left( \frac{A_j P_{V_j} x}{\|A_j P_{V_j} x\|} \right) = 1$. Thus, it follows that with probability $1 - \delta/6$, we have that the hypothesis we output, $h_i$, makes no errors on points where $g_i(x) = 1$, as desired.

To complete the proof of the theorem, it now remains to show that \learnh{} is respectful. Towards this, note that all queries are made in our calls to \learnmargin{}. Notably, we query $f(P_{V_i}^T A_i^{-1} x)$ on points drawn from $\calN(0, \frac{\pi}{2} I_{V_i})$. Since $R_i \cap V_i \not = \emptyset$, it follows that there must exist a point $z \in V_i$ such that $z_n \not = 0$. Thus, querying $f$ on a point of the form $(x,0)$ is a measure zero event and \learnh{} is respectful as desired. Using \Cref{lem:respectful} then yields the result.
\end{proof}

\section*{Discussion}

We have shown that PQ and TDS learning are equivalent in the distribution-free setting, where learners are required to succeed for arbitrary training distributions. A natural open question is whether this equivalence extends to the distribution-specific setting, in which PQ or TDS learners are only required to succeed when the training distribution equals a fixed reference distribution $\calD^*$, while no assumptions are made on the test distribution.

Existing work in distribution-specific learning establishes positive results for both models, but does not always yield matching guarantees. For example, there exists an efficient TDS learner for $\mathsf{AC}^0$ circuits when $\calD^*$ is the uniform distribution on the Boolean hypercube, whereas no efficient PQ learner is currently known for this setting. Nevertheless, no formal separation between the two models has been established.

\begin{flushleft}
\bibliographystyle{alpha}
\bibliography{allrefs}
\end{flushleft}

\newpage
\appendix
\section{Useful Facts from Probability Theory}
\begin{lemma}[Reverse Markov Inequality]
Suppose that $\bx$ is a random variable satisfying $\bx \leq B$ and $\E[\bx] \geq 0$, then
	\[\Pr \left[ \bx \geq \frac{1}{2} \E[\bx] \right] \geq \frac{1}{2B} \E[\bx] \]	
\end{lemma}

\begin{proof}
Note that	
\[\E[\bx] \leq \Pr \left[ \bx \geq \frac{1}{2} \E[\bx] \right] \cdot B + \frac{1}{2} \E[\bx].\]
Rearranging yields the result.
\end{proof}

We will also need the VC inequality.

\begin{theorem}[VC Inequality]
    Let $\calC$ denote a concept class over $\calX$, $\calD$ be a distribution over $\calX$, and $\eps \in (0,1)$. If $d$ is the VC dimension of $\calC$ and $\bS \sim \calD$ is a set of size $m$, then
        \[\Pr_{\bS} \left[ \sup_{c \in \calC} \left| \Pr_{\bx \sim \calD}[c(\bx) = 1] - \frac{1}{|\bS|} \sum_{\bs \in \bS} \mathbb{I}\{c(\bs) = 1\} \right| \geq \eps \right] \leq 8 \cdot \left( \frac{e m}{d} \right)^d \cdot e^{-m \eps^2/32} \]
\end{theorem}

\begin{corollary}
\label[corollary]{cor:VC-sampling}
    Let $\calC$ denote a concept class over $\calX$, $\calD$ be a distribution over $\calX$, and $\eps, \delta \in (0,1/3)$. If $d$ is the VC dimension of $\calC$ and $\bS \sim \calD$ is a set of size $\Omega \left( \frac{d \log(1/\eps \delta)}{\eps^2} \right)$, then with probability $1 - \delta$, 
        \[\left| \Pr_{\bx \sim \calD}[c(\bx) = 1] - \frac{1}{m} \sum_{s \in S} \mathbb{I}\{c(\bs) = 1\} \right| \leq \eps\]
    for all $c \in \calC$.
\end{corollary}

\section{Halfspaces and Forster Transform}
Recall that we call a halfspace homogeneous if is of the form $\sgn(w \cdot x)$, i.e. it goes through the origin. A crucial notion will be 

\begin{definition}[Margin]
Given a halfspace $f(x) = \sgn(w \cdot x)$ and a point $x$, we say that the margin of $x$ with respect to $f$ is 
    \[\frac{|w \cdot x|}{\|w\|_2 \|x\|_2}\]
\end{definition}

Another notion we'll need is

\begin{definition}
\label[definition]{def:margin}
	We say that a set of points $S \subseteq \mathbb{S}^{n-1}$ is in $\eps$-approximate radial isotropic position if
		\[\frac{1-\eps}{n} I_n \preceq \frac{1}{|S|} \sum_{x \in S} xx^T \preceq \frac{1+\eps}{n} I_n   \]
\end{definition}

Notably, such a condition is useful as it implies an anti-concentration bound for any halfspace. In particular,

\begin{lemma}
\label[lemma]{lem:forster-margin}
Suppose that $S \subseteq \mathbb{S}^{n-1}$ is a set of points in $1/2$-approximate radial isotropic position and $w \in \mathbb{S}^{n-1}$, then 
	\[\Pr_{\bx \sim S} \left [|w \cdot \bx| \geq \frac{1}{2 \sqrt{n}} \right] \geq \frac{1}{4n}  \]	
\end{lemma}

\begin{proof}
Note that
	\[\frac{1}{2n} \leq  w^T \E_{\bx \sim S}[\bx\bx^\top] w = \E_{\bx \sim S}[(w \cdot \bx)^2]\]
Since $(w \cdot \bx)^2 \leq \|w\|_2^2 \|\bx\|_2^2 = 1$ by Cauchy Schwartz, the result follows from the Reverse Markov inequality applied to $(w \cdot \bx)^2$.
\end{proof}

Crucially, there exists efficient algorithms to place points into $\eps$-radial isotropic position.

\begin{theorem}[\cite{diakonikolas2023strongly}]
\label[theorem]{thm:forster}
There exists a randomized algorithm that given a set $S \subseteq \R^n$ of size $m$ and an $\eps \in (0,1)$, runs in time strongly polynomial in $mn/\eps$ and has the following high probability guarantee: either the algorithm finds a matrix $A$ such that $\{Ax/\|Ax\|: x \in S\}$ is in $\eps$-approximate radial isotropic position, or it finds a proper subspace $W \subsetneq \R^n$ such that $|S \cap W| > (\dim(W)/n) \cdot |S|$.
\end{theorem}

\begin{corollary}
\label[corollary]{cor:forster}
	Given a set of points $S \subseteq \R^n$ and parameter $\delta \in (0,1)$, there exists an algorithm $\forster(S,\delta)$ that outputs a subspace $V$ and matrix $A: V \rightarrow V$ such that $\{A P_V x/\|A P_V x\|: x \in V\}$ is in $\eps$-approximate radial isotropic position in $\mathbb{R}^{\dim(V)}$, where $P_V: \mathbb{R}^n \rightarrow V$ denotes the projection matrix onto $V$. Moreover, the algorithm runs in time strongly polynomial in $mn \log(1/\delta)$ and $|S \cap V| > (\dim(V)/n) \cdot |S|$.
\end{corollary}

\begin{proof}
We now start by noting that we can boost the failure probability of the algorithm in \Cref{thm:forster} to $\delta/n$ by repeating it $\log(n/\delta)$ times and checking whether the points are in radial isotropic condition or the subspace $W$ contains enough points. As such, we'll assume throughout that it has failure probability $\delta/n$.

Our procedure is then the following: Let $V = \R^n$ and $S' = S$. Loop at most $n$ times: Run the algorithm from \Cref{thm:forster} on $S' \subseteq V$ with $\eps = 1/2$ to either get a matrix $A$ or subspace $W \subsetneq V$. If we get a matrix, output $A,V$. Otherwise, update $V \gets W$ and $S' \gets P_W (S' \cap W)$.

We claim that if the first $n$ calls to the algorithm in \Cref{thm:forster} all succeed, then the above algorithm satisfies the corollary. Indeed, the process will eventually return a matrix since the dimension of $V$ decreases in each iteration. Thus, we indeed have that $\{A P_V x/ \|A P_V x\|: x \in S \cap V\}$ is in $1/2$-radial isotropic position. Moreover, if $ \R^n := W_0 \supseteq W_1 \supseteq W_2 \dots \supseteq W_\ell$ are the subspaces in each iteration of the loop, then
	\[|W_\ell \cap S| > \frac{\dim(W_\ell)}{\dim(W_{\ell-1})} |W_{\ell-1} \cap S| \geq |S| \prod_{i=1}^{\ell} \frac{\dim(W_i)}{\dim(W_{i-1})} = \frac{\dim(W_\ell)}{n} \cdot |S|.\]
Finally, the statement on the runtime follows from the polynomial runtime of \Cref{thm:forster}.
\end{proof}

\begin{remark}
    Note that applying a Forster transform preserves the property of a point being labeled by a homogeneous halfspace. In particular, let $S \subseteq \R^n$ and suppose that we are given labeled pairs $(x,y)$, where $x \in S$ and $y = \sgn(w \cdot x)$ for some $w \in \R^n$. We then observe that for $x \in V$ we have that the labeled point set $\left( \frac{A P_V x}{\|A P_V x\|}, y\right)$, where $y = \sign(w \cdot x)$, are labeled by a halfspace in the transformed space, namely $\sgn \left ( (A^{-1})^T P_V w \cdot \frac{A_i P_V x}{\|A_i P_V x\|} \right)$.
\end{remark}

\section{The Agnostic Setting}\label{section:agnostic}

In our main theorem (\Cref{theorem:pq-via-tds-main}), we have shown that PQ learning can be reduced to TDS learning. However, the premise of our result requires that the training and test distributions are both realizable by the same function $f\in\calC$. We will now focus on the agnostic setting where neither the training not the test labels need be realizable.

\begin{definition}[Agnostic Setting \citep{ben2006analysis,blitzer2007learning}]\label[definition]{definition:agnostic}
    Let $\calC$ be a concept class over $\calX$. In the agnostic setting, the learner has labeled sample access to a training distribution $\dtrainlabeled$, and unlabeled sample access to the $\calX$-marginal of a test distribution $\dtestlabeled$, and the goal is to produce a hypothesis whose  on the test distribution competes with the following benchmark:
    \[
        \lambda = \lambda(\dtrainlabeled,\dtestlabeled;\calC) := \min_{f\in\calC} \Bigr(\Pr_{(\bx,y)\sim \dtrainlabeled}\bigr[ f(\bx)\neq y\bigr] + \Pr_{(\bx,y)\sim \dtestlabeled}\bigr[ f(\bx)\neq y\bigr] \Bigr)
    \]
\end{definition}

The parameter $\lambda$ quantifies the relationship between the training and test labels and, in the absence of test labels, a dependence of the output's test error on $\lambda$ is unavoidable \citep{klivans2024testable}. We may now define the agnostic versions of TDS and PQ learning.

\begin{definition}[Agnostic TDS Learning \citep{klivans2024testable}]
\label[definition]{definition:agnostic-tds-learning}
    Let $\calC$ be a concept class over $\calX$. The algorithm $\calA$ is a TDS learner for $\calC$ up to error $O(\lambda)+\eps$ and probability of failure $\delta$ in the agnostic setting if it either outputs $\rej$ or outputs $\acc$ and returns a hypothesis $h:\calX\to \{\pm 1\}$ such that, with probability at least $1-\delta$, the following hold:
    \begin{enumerate}[label=\textnormal{(}\alph*\textnormal{)},leftmargin=6ex]
    \item\label{item:soundness-tds-agnostic} (\textit{soundness}) Upon acceptance, the test error is $\Pr_{(\bx,\by)\sim \dtestlabeled}[\by \neq h(\bx)] \le A\cdot\lambda+ \accuracy$, where $A\ge 1$ is some constant.
    \item\label{item:completeness-tds-agnostic} (\textit{completeness}) 
    If $\dtrain =\dtest$, then the algorithm accepts.
    \end{enumerate}
\end{definition}

\begin{definition}[Agnostic PQ Learning \citep{goldwasser2020beyond}]
\label[definition]{definition:agnostic-pq-learning}
    Let $\calC$ be a concept class over $\calX$. The algorithm $\calA$ is a PQ learner for $\calC$ up to error $O(\lambda)+\eps$, rejection rate $\rejrate$ and probability of failure $\delta$ in the agnostic setting if it outputs, with probability at least $1-\delta$, a hypothesis $h:\calX\to \{\pm 1\}$ and a selector $g:\calX\to \{0,1\}$ such that:
    \begin{enumerate}[label=\textnormal{(}\alph*\textnormal{)},leftmargin=6ex]
    \item\label{item:accuracy-pq-agnostic} (\textit{accuracy}) The test error is bounded as $\Pr_{\bx\sim \dtest}[f(\bx) \neq h(\bx)\text{ and }g(\bx) = 1] \le A\cdot \lambda+\accuracy$, where $A\ge 1$ is some constant.
    \item\label{item:rejrate-pq-agnostic} (\textit{rejection rate}) The probability of rejection is bounded as $\Pr_{\bx\sim \dtrain}[g(\bx) = 0] \le \rejrate$.
    \end{enumerate}
\end{definition}

The following theorem states that efficient agnostic TDS learning up to error $O(\lambda)$ implies efficient agnostic PQ learning with error rate $O(\lambda/\eta)$ and rejection rate $O(\eta)$, where $\eta$ can be chosen freely. Note that no PQ learner can attain error better than $\Omega(\lambda/\eta)$ and rejection rate at most $\eta$ even if given an unbounded number of samples and runtime. (This follows from a straightforward modification of Lemma D.1 in \cite{goldwasser2020beyond}).

\begin{theorem}[PQ learning via TDS learning: The Agnostic Case]\label{theorem:pq-via-tds-agnostic}
    Let $\calA$ be an agnostic TDS learner for some class $\calC$ that achieves error $O(\lambda)+\eps$, runs in time $T_\calA(\eps)$ and has total sample complexity $m_\calA(\eps)$. Then, there is a PQ learner for $\calC$ that receives the parameters $(\eps,\eta,\delta)$ as inputs and achieves error $O(\frac{\lambda}{\eta} + \eps)$, rejection rate $O(\eta+\eps)$ and probability of failure $\delta$. In particular, the algorithm runs in time $T_\calA \cdot \widetilde{O}_{\eps,\delta}(m_\calA^6) + \poly(m_\calA,\frac{1}{\eps},\frac{1}{\eta},\log \frac{1}{\delta})$ and uses $\poly(m_\calA,1/\eps,1/\eta,\log 1/\delta)$ samples from the training and test distributions.
\end{theorem}

\begin{proof}
    The algorithm for the agnostic case is very similar to \Cref{algorithm:boosting}, with one modification. In particular, after computing the estimates $\hat p_{i,t}^b$ up to error $\gamma_1 = \frac{\eps\eta}{4 T (T+1)}$ and ensuring that they are both at least equal to $\pmin$, the algorithm checks whether $\hat p_{i,t}^1 \le \eta\cdot \hat p_{i,t}^0$, in which case the node $(i,t)$ becomes a leaf node with $\lbl(i,t) = 0$.

    Recall that in the realizable case we had three types of leaf nodes of the branching program $P$ constructed by \Cref{algorithm:boosting} (see \Cref{item:first-type-leaves,item:second-type-leaves,item:third-type-leaves}). The modified algorithm described above creates one additional type of leaf nodes $\calT_4\subseteq\leaves(P)$ for which we have ensured that $\hat p_{N}^1 \le \eta\cdot \hat p_{N}^0$. For the following, suppose that we are in the event that holds with probability at least $1-O(\delta)$ where the estimates $\hat p_{i,t}^b$ are $\gamma_1$-close to their expected values $p_{i,t}^b$.
    
    Observe that the distribution $\Dtrain$ is a mixture distribution with components $\{\calG_{N}^1\}_{N\in\leaves(P)}$. Each component $\calG_{N}^1$ corresponds to weight $p_{N}^1$. Similarly, $\Dtest$ is a mixture of $\{\calG_{N}^0\}_{N\in\leaves(P)}$ with weights $\{p_N^0\}_{N\in\leaves(P)}$.

    \paragraph{Rejection Rate} For the rejection rate, the only difference between the agnostic case and the proof of \Cref{theorem:pq-via-tds-main} is the existence of nodes $N$ of the new type $\calT_4$, where we have $\hat p_{N}^1 \le \eta\cdot \hat p_{N}^0$. The additional term of the rejection rate that corresponds to these nodes is bounded by the following quantity:
    \[
        \sum_{N\in \calT_4} p_{N}^1 \le O(\gamma_1 T^2) + \sum_{N\in \calT_4} \hat p_{N}^1 \le O(\gamma_1 T^2) + \eta\cdot \sum_{N\in \calT_4}\hat p_{N}^0 \le O(\gamma_1 T^2) + \eta\cdot \sum_{N\in\leaves(P)} p_N^0 \le O(\eta)
    \]
    Therefore, the overall rejection rate is $O(\eta+\eps)$.

    \paragraph{Error Rate} For the error rate on the test distribution, we focus on the nodes of type $\calT_3$ (see \Cref{item:third-type-leaves}), i.e., the nodes $N$ where the TDS learner has produced hypotheses $h_N$ for which:
    \[
        \Pr_{(\bx,\by)\sim\bar\calG_N^0, P}\bigr[h_N(\bx)\neq \by \bigr] \le A\cdot \underbrace{\min_{f\in \calC}\Bigr( \Pr_{(\bx,\by)\sim \bar\calG_N^1}\bigr[ f(\bx)\neq \by\bigr] + \Pr_{(\bx,\by)\sim \bar\calG_N^0}\bigr[ f(\bx)\neq \by\bigr] \Bigr)}_{\lambda_{N}} \,+\, \eps
    \]
    The total error incurred by these nodes is bounded as follows by using the fact that $\calT_3\cap\calT_4 = \emptyset$ and therefore for any $N\in\calT_3$ we have $\hat p_N^0<\frac{\hat p_N^1}{\eta}$.
    \begin{align*}
        &\sum_{N\in \calT_3} p_N^0\cdot \Pr_{(\bx,\by)\sim\bar\calG_N^0, P}\bigr[h_N(\bx)\neq \by \bigr] \le O(\eps)+\sum_{N\in \calT_3} \hat p_N^0\cdot \Pr_{(\bx,\by)\sim\bar\calG_N^0, P}\bigr[h_N(\bx)\neq \by \bigr] \\
        &\le A\cdot {\min_{f\in \calC} \sum_{N\in\calT_3}\biggr( \hat p_N^0\cdot \Pr_{(\bx,\by)\sim \bar\calG_N^1}\bigr[ f(\bx)\neq \by\bigr] + \frac{\hat p_N^1}{\eta}\cdot \Pr_{(\bx,\by)\sim \bar\calG_N^0}\bigr[ f(\bx)\neq \by\bigr] \biggr)} \,+\, O(\eps) \\
        &\le \frac{A}{\eta}\cdot {\min_{f\in \calC} \sum_{N\in\calT_3}\biggr(  p_N^0\cdot \Pr_{(\bx,\by)\sim \bar\calG_N^1}\bigr[ f(\bx)\neq \by\bigr] + { p_N^1}\cdot \Pr_{(\bx,\by)\sim \bar\calG_N^0}\bigr[ f(\bx)\neq \by\bigr] \biggr)} \,+\, O(\eps) \\
        &\le \frac{A}{\eta}\cdot {\min_{f\in \calC} \sum_{N\in\leaves(P)}\biggr(  p_N^0\cdot \Pr_{(\bx,\by)\sim \bar\calG_N^1}\bigr[ f(\bx)\neq \by\bigr] + { p_N^1}\cdot \Pr_{(\bx,\by)\sim \bar\calG_N^0}\bigr[ f(\bx)\neq \by\bigr] \biggr)} \,+\, O(\eps) \\
        &\le \frac{A\cdot \lambda}{\eta} + O(\eps)\,.
    \end{align*}
    The error incurred by the leaf nodes that lie outside $\calT_3$ remains unchanged, which concludes the proof of \Cref{theorem:pq-via-tds-agnostic}.
\end{proof}

\begin{remark}
\cite{goldwasser2020beyond} showed that access to an agnostic ERM oracle in the standard (i.e., no distribution shift) setting suffices to construct an agnostic PQ learner using only a polynomial number of oracle calls. However, our result is not implied by theirs for the following reasons:
\begin{enumerate}
    \item Upon acceptance, our TDS learner is allowed to output a constant-factor approximate hypothesis, whereas the reduction of \cite{goldwasser2020beyond} requires a hypothesis achieving optimal error.
    \item Their approach further requires the learning algorithm to be proper (i.e., an ERM), whereas our reduction does not impose such a requirement.
\end{enumerate}
\end{remark}

\end{document}